\documentclass[runningheads]{llncs}

\usepackage{amsmath}
\usepackage{latexsym}
\usepackage{amsfonts,amssymb,amsmath}
\usepackage{graphicx, gastex}
\usepackage{cite}
\usepackage{mathrsfs}

\DeclareSymbolFont{rsfscript}{OMS}{rsfs}{m}{n}
\DeclareSymbolFontAlphabet{\mathrsfs}{rsfscript}

\DeclareMathOperator{\dt}{.}
\DeclareMathOperator{\Syn}{Syn}

\newcommand{\Fact}{\textit{Fact}}%{\operatorname{Fact}}
\begin{document}
\title{Finitely generated ideal languages \\ and synchronizing automata}
\author{Vladimir V. Gusev, Marina I. Maslennikova, Elena V. Pribavkina}

\institute{Ural Federal University, Ekaterinburg, Russia\\
\email{vl.gusev@gmail.com, maslennikova.marina@gmail.com, elena.pribavkina@usu.ru}}

\maketitle

\begin{abstract}
%MSA is not strongly connected
%study sc synch
%constructions from comb words
%
We study representations of ideal languages by means of strongly
connected synchronizing automata. For every finitely generated ideal
language $L$ we construct such an automaton with at most $2^n$ states,
where $n$ is the maximal length of words in $L$. Our constructions are based
on the De Bruijn graph.\\
%
%We study finitely generated ideal languages as languages of synchronizing words. We provide an algorithm to construct a strongly connected automaton for which such a language serves as the language of synchronizing words.\\
\end{abstract}

\section{Introduction}
Let $\mathscr{A}=\langle Q,\Sigma,\delta\rangle$ be a \textit{deterministic finite automaton} (DFA for short),
where $Q$ is the \textit{state set}, $\Sigma$ stands for the \textit{input alphabet},
and $\delta: Q\times\Sigma\rightarrow Q$ is the \textit{transition function} defining
an action of the letters in $\Sigma$ on $Q$.
%The action extends in a natural way to an action $Q\times \Sigma^{*}\rightarrow Q$ of
%the free monoid $\Sigma^{*}$ over $\Sigma$; the latter action is also denoted by $\delta$.
When $\delta$ is clear from the context, we will write $q\dt w$ instead
of $\delta(q,w)$ for $q\in Q$ and $w\in \Sigma^{*}$.
%In the theory of formal languages the definition of a DFA usually includes the
%set $F\subseteq Q$ of \textit{terminal states} and an \textit{initial state} $q_0\in Q$.
%We will use this definition when dealing with automata as devices for recognizing languages.
%The language $L\subseteq \Sigma^{*}$ is \textit{recognized} (or \textit{accepted}) by an automaton $\mathscr{A}=\langle Q,\Sigma,\delta, F, q_0\rangle$ if $L=\{w\in\Sigma^{*}\mid \delta(q_0,w)\in F\}$. We also use standard concepts of the theory of formal languages such as regular language, minimal automaton, etc. \cite{Perrin}

A DFA $\mathscr{A}=\langle Q,\Sigma,\delta\rangle$ is called
\emph{synchronizing} if there exists a word $w  \in \Sigma^{*}$ which leaves
the automaton in unique state no matter at which state in $Q$ it is applied: $q\dt w=q'\dt w$
for all $q, q' \in Q$. Any word $w$ with such property is said to be
\textit{synchronizing} (or \emph{reset}) word for the DFA $\mathscr{A}$.
%This notion has been actively studied
%since the work of Jan \v{C}ern\'{y} \cite{Ce64} in 1964.
For the last 50 years synchronizing automata received a great deal of attention.
For a brief introduction to the theory of synchronizing automata we refer the reader
to the recent surveys~\cite{Vo_Survey, Sa05}.
%For motivation and applications we refer the reader
%to the recent surveys~\cite{Vo_Survey, Sa05}.

In the present paper we focus on language theoretic aspects of the theory of synchronizing automata.
We denote by $\Syn(\mathrsfs{A})$ the language of synchronizing words for a given automaton $\mathscr{A}$.
It is well known that $\Syn(\mathrsfs{A})$ is regular~\cite{Vo_Survey}.
%Furthermore, it is a \emph{two-sided ideal}
Furthermore, it is an \emph{ideal}
%(or simply \emph{ideal})
in $\Sigma^*$, i.e. $\Syn(\mathrsfs{A})=\Sigma^{*}\Syn(\mathrsfs{A})\Sigma^{*}$.
On the other hand, every ideal language $L$ serves as a language of synchronizing words for some automaton.
%For example, $\Syn(\mathrsfs{A}_L) = L$, where $\mathrsfs{A}_L$ is the minimal automaton of the language $L$.
For instance, the minimal automaton
%$\mathrsfs{A}_L$
of the language $L$ is synchronized by $L$~\cite{SOFSEM}.
%of $L$ has $L$ as a language of synchronizing words~\cite{SOFSEM}.
%It was observed in \cite{SOFSEM} that the minimal deterministic automaton $\mathrsfs{A}_L$ recognizing an
%ideal regular language $L$ is synchronizing, and $\Syn(\mathrsfs{A}_L)=L$.
Thus, synchronizing automata can be considered as a special representation of ideal languages.
%Thus, we may consider synchronizing automata as a special representation of ideal languages.
%It was noted in~\cite{SOFSEM} that such representation could be extremely effective.
Effectiveness of such representation was addressed in~\cite{SOFSEM}.
%It was proved in~\cite{SOFSEM} that such representation could be
%exponentially smaller than usual representation via minimal automaton of $L$.
%Comparison between traditional minimal(in general regular languages are represented via minimal automaton).
The \textit{reset complexity} $rc(L)$ of an ideal language $L$ is the minimal possible number
of states in a synchronizing automaton $\mathscr{A}$ such that $\Syn(\mathscr{A})=L$.
Every such automaton $\mathscr{A}$ is called \textit{minimal synchronizing automaton} (for brevity, MSA).
%In some cases it could be much more (excponentially) effective than traditional representation as a minimal automaton.
%Usually, regular languages are represented via minimal automaton. Number of states is important.
Let $sc(L)$ be the number of states in the minimal automaton recognizing $L$.
%Then there are languages
%Often $sc(L)$
%is called state complexity of a language $L$.
Then for every ideal language $L$ we have  $rc(L) \leq sc(L)$\footnote{since the minimal automaton is synchronized by $L$}.
Moreover, there are languages $L_n$ for every $n \geq 3$ such that $rc(L_n) = n$ and $sc(L_n) = 2^n - n$, see~\cite{SOFSEM}.
Thus, representation of an ideal language by means of a synchronizing automaton can be exponentially smaller than
``traditional'' representation via minimal automaton.
However,
no reasonable algorithm is known for computing MSA of a given language.
One of the obstacles is that MSA is not uniquely defined.
%by a language.
For instance, there is a language with at least two different
MSA's: one of them is strongly connected, another one has a sink state~\cite{SOFSEM}.
Therefore, some refinement of the notion of MSA seems to be necessary.
Another important observation is the following:
minimal synchronizing automata for the aforementioned languages $L_n$ are
strongly connected. Thus, one may expect that there is always a strongly connected
MSA for an ideal language.
In the present paper
% ( In section~\ref{sigma-n} )
we show that it is not the case.
Moreover, the smallest strongly connected automaton with a language $L$ as the
language of synchronizing words may be exponentially larger than %the minimal automaton of $L$.
a minimal synchronizing automaton of $L$.

Another source of motivation for studying representations of ideal languages by means
of synchronizing automata comes from the famous \v{C}ern\'{y} conjecture.
% In 1964
\v{C}ern\'{y}
already
in 1964
%he
conjectured that
every synchronizing automaton possesses a synchronizing word of length at most $(n-1)^2$.
Despite intensive efforts of researchers this conjecture is still widely open.
We can
restate the
%it
\v{C}ern\'{y} conjecture
%can be restated
%reformulate the \v{C}ern\'{y} conjecture
%can be stated
in terms of reset complexity as follows:
if $\ell$ is the minimal length of words in
an ideal language $L$ then $rc(L)\ge \sqrt{\ell}+1$.
Thus, we hope that deeper understanding of reset complexity will
bring us new ideas to resolve this long standing conjecture.
It is well known that the \v{C}ern\'{y} conjecture holds true whenever
it holds true for strongly connected automata.
In this regard an interesting related question
%has arisen~\cite{PrincId}:
was posed in~\cite{PrincId}:
does every ideal language
serve as the language of synchronizing words for some strongly connected automaton?
For instance, if the answer is negative then there is a way to simplify formal language statement
of the \v{C}ern\'{y} conjecture. Unfortunately, it is not the case. Recently
Reiss and Rodaro~\cite{ReRo13}
for every ideal language\footnote{over an alphabet with at least two letters} $L$ presented a strongly connected automaton $\mathrsfs{A}$ such that
$Syn(\mathrsfs{A}) = L$.
%presented for
%proved that every ideal language can be represented by means of a strongly connected automaton.
%have presented strongly connected
%synchronizing automaton with a given language as a language of synchronizing words.
Their proof is non-trivial and technical.
In the present paper we give simple constructive proof of the fact that every finitely
generated ideal language $L$, i.e. $L = \Sigma^*U\Sigma^*$ for some finite set $U$,
serves as the language of synchronizing words of some strongly connected automaton.
Our constructions reveal interesting connections with classical objects from combinatorics
on words.

\section{Algorithms and automata constructions}

Let $\Sigma$ be a finite alphabet with $|\Sigma|>1$.
Let $L$ be a finitely generated ideal language over $\Sigma$, i.e. $L=\Sigma^*S\Sigma^*$, where $S$ is a finite set of words.
In this section we construct a strongly connected synchronizing automaton for which $L=\Sigma^*S\Sigma^*$.

First recall some standard definitions and fix notation. A word $u$ is a \emph{factor} (\emph{prefix}, \emph{suffix})
of a word $w$, if $w=xuy$ ($w=uy$, $w=xu$ respectively) for some $x,y\in\Sigma^*$. By $\Fact(w)$ we denote
the set of all factors of $w$. The $i^{th}$ letter of the word $w$ is denoted by $w[i]$. The factor $w[i]w[i+1]\cdots w[j]$
is denoted by $w[i..j]$. By $\Sigma^n$ ($\Sigma^{\le n}$, $\Sigma^{\ge n}$) we denote the set of all words over $\Sigma$ of length $n$
(at most $n$, at least $n$ respectively).

Note, that if a word $s\in S$ is a factor of some other word $t\in S$,
then the word $t$ may be deleted from the set $S$ without affecting
the ideal language, generated by $S$. Thus, we may assume, that the set
$S$ is \emph{anti-factorial}, i.e.\ no word in $S$ is a factor of another
word in $S$.

\subsection{Ideal language generated by $\Sigma^n$}

\begin{theorem}
Let $\Sigma = \{a,b\}$.
% and $L=\Sigma^{\ge n}$.
There is unique up to isomorphism strongly connected synchronizing automaton $\mathrsfs{B}$
such that $\Syn(\mathrsfs{B}) = \Sigma^{\ge n}$.
%Let $L=\Sigma^{\ge n}$, $|\Sigma|>1$. Then $\rcs(L)=2^n$.
\end{theorem}

\begin{proof}
%Assume first $\Sigma=\{a,b\}$.
%Fix a word $w$ over $\Sigma=\{a,b\}$. Let $|w|=n$. Denote the $i$-th letter of $w$ by $w[i]$ and the subword $w[i]w[2]...w[j]$ by $w[i..j]$.
%Consider De Bruijn automaton for the words of length $n$, i.e.\ automaton $\mathscr{B}=\right<\Sigma^n,\Sigma,\delta\left>$, where $\delta(xu,y)=uy$
%for $u\in\Sigma^{n-1}$, $x,y\in\Sigma$.
Consider De Bruijn graph for the words of length $n$. Recall that the vertices of this graph are the words of length $n$, and there is a directed edge from the vertex $u$ to the vertex $v$, if $u=xs$ and $v=sy$ for some $s\in\Sigma^{n-1}$, $x,y\in\Sigma$. By labeling each edge $e=(u,v)$ by the last letter of $v$
we obtain De Bruijn automaton. Its state set is $Q=\Sigma^n$, and transition function is defined in the following way: $xs\dt y=sy$ for $s\in\Sigma^{n-1}$, $x,y\in\Sigma$.
%one classical construction from the theory of formal languages which belongs to De Bruijn. Take the set $\Sigma^n$ including all words over $\Sigma$ of length $n$. We refer to the set $\Sigma^n$ as to the state set $Q$ of a DFA $\mathscr{B}$. The transition functions is defined as follows: $\delta(w,x)=w[2..n]x$, where $x\in \Sigma$.
%The corresponding construction for $n=3$ is shown on the Fig.\ref{DeBr3}.
%\begin{figure}[ht]
%\begin{center}
%  \begin{picture}(60,48)
%   \gasset{Nadjust=w,Nadjustdist=1.5,Nh=6,Nmr=3}
%   \thinlines
%   \node[Nw=9](aaa)(0,35){$aaa$}
%   \node[Nw=9](aab)(20,50){$aab$}
%   \node[Nw=9](aba)(20,35){$aba$}
%   \node[Nw=9](abb)(40,50){$abb$}
%   \node[Nw=9](baa)(0,0){$baa$}
%   \node[Nw=9](bab)(20,20){$bab$}
%   \node[Nw=9](bba)(20,0){$bba$}
%   \node[Nw=9](bbb)(40,20){$bbb$}
%   \drawloop[loopdiam=6,loopangle=90](aaa){$a$}
%   \drawloop[loopdiam=6,loopangle=0](bbb){$b$}
%   \drawedge(aaa,aab){$b$}
%   \drawedge(aab,aba){$a$}
%   \drawedge(aab,abb){$b$}
%   \drawedge[curvedepth=-2,sxo=-2](aba,baa){$a$}
%   \drawedge(aba,bab){$b$}
%   \drawedge[curvedepth=5,sxo=-4](baa,aaa){$a$}
%   \drawedge[curvedepth=3,sxo=-2](baa,aab){$b$}
%  \drawedge[sxo=1](abb,bba){$a$}
%   \drawedge[curvedepth=5,sxo=4](abb,bbb){$b$}
%   \drawedge[curvedepth=3](bab,aba){$a$}
%   \drawedge(bab,abb){$b$}
%   \drawedge(bbb,bba){$a$}
%   \drawedge(bba,baa){$a$}
%   \drawedge(bba,bab){$b$}
%   \end{picture}
%\end{center}
%\caption{De Bruijn automaton for $n=3$}
%\label{DeBr3}
%\end{figure}
De Bruijn automaton is known to be strongly connected. Thus it remains to verify that $\Syn(\mathscr{B})=L$.
It is easy to see that for an arbitrary word $u$ of length at most $n$ we have $Q\dt u=\Sigma^{n-|u|}u$.
Hence for any word $w$ of length $n$ we have $|Q\dt w|=1$, and for any word $u$ of length less than $n$ we have $|Q\dt u|>1$.
So, $\Syn(\mathrsfs{B})=L.$
%It can be easily seen that each word of length $n$ synchronizes $\mathscr{B}$. Indeed, applying a letter $x$ to an arbitrary state results in the
%state whose last letter is $x$. In other words, we have $Q\dt x=\Sigma^{n-1}x$. Thus, by induction we have $Q\dt u =\Sigma^{n-|u|}u$ for an arbitrary word $u$ such that $|u|\le n$. Hence $|Q\dt w|=1$ for any word $w\in\Sigma^n$.
% Now it remains to check that there is no synchronizing words of length less than $n$. %Consider subsets $Q\dt a=\Sigma^{n-1}a$ and $Q\dt b=\Sigma^{n-1}b$.
 %These subsets include all words of length $n$ that end up with $a$ and $b$ respectively.
% We have $|Q\dt a|=|Q\dt b|=2^{n-1}$ and $Q=Q\dt a \cup Q\dt b$. Take $u\dt \Sigma^k$, where $k<n$. By induction we obtain that $Q\dt u$ is the set of all words of length $n$ whose suffix is $u$. Thus $|Q\dt u|=2^{n-|u|}$ and $Q\dt u\neq Q\dt v$ if $u\neq v$. So $Q=\bigcup\limits_{u\in \Sigma^k} Q\dt u$. Hence there is no synchronizing words of length less than $n$.

%Further we prove that $rcs(L)=2^n$. We provide a new technique that will be used in the present paper later. Again we explain the technique for binary alphabet, that is $\Sigma=\{a,b\}$. In general case all arguments are similar.
%Now we prove that $\rcs(L)=2^n$. Actually we prove even more, namely, that De Bruijn automaton is unique up to isomorphism
%strongly connected automaton whose language of synchronizing words coincides with $L$.
Let $\mathscr{C}=\langle Q,\Sigma,\delta\rangle$ be a strongly connected synchronizing DFA such that $\Syn(\mathscr{C})=L$. Let us prove that $|Q|\le 2^n$. Strong connectivity implies $Q\dt a\cup Q\dt b=Q$.
By induction it is easy to see that $Q=\bigcup_{|w|=k}Q\dt w$. In particular, we have
$Q=\bigcup_{|w|=n}Q\dt w$. Thus, $|Q|=|\bigcup_{|w|=n}Q\dt w|\le \sum_{|w|=n}|Q\dt w|=2^n$.
The last equality follows from the fact that every word of length $n$ synchronizes $\mathrsfs{C}$, so each $Q\dt w$ is a singleton.
For the converse inequality $2^n\le|Q|$ consider the DFA $\mathscr{C}_a$,
%Consider the action of the letter $a$ on the state set $Q$.
%Denote by $\mathscr{C}_a$ the DFA
obtained from $\mathscr{C}$ by removing all transitions corresponding to the action of $b$ in $\mathscr{C}$. The word $a^n$ synchronizes $\mathscr{C}$, so $\mathscr{C}_a$ contains no cycles but unique loop. %Otherwise we have $p\dt a^n\neq q\dt a^n$ for two different states $p$ and $q$ that appeared in some cycle in $\mathscr{C}_a$. Let $Q\dt a^n={s}$, that is $p\dt a^n=s$ for all $p\in Q$. Note that any word $a^{l}$ with $l<n$ is not synchronizing for $\mathscr{C}$. Then there exists a state $p_1\in Q$ such that $p_1\dt a^n=s$ and $p_1\dt a^l \neq s$ for all $l<n$. Finally, $s\dt a=s$. Otherwise a cycle marked by $a$ appears in $\mathscr{C}_a$.
So the automaton $\mathscr{C}_a$ has a tree-like structure as it is shown on Fig.\ref{Ca}.
Denote by $s$ the state of $\mathrsfs{C}$ such that $s\dt a=s$. The state $s$ is called \emph{root}
of the tree, and the states $p_1$, $p_2$, $\ldots,$ $p_k$ having no incoming transitions labeled by $a$ are called \emph{leaves}
of the tree. The \emph{height} $h(p_i)$ of a vertex $p_i$ is the length of the path from $p_i$ to the root $s$.
The height of the tree $h(\mathrsfs{C}_a)$ is the maximal height of its leaves. We have $h(\mathrsfs{C}_a)=n$. Indeed,
if $h(\mathrsfs{C}_a)=h<n$, then we would have $Q\dt a^h=\{s\}$, meaning that $a^h\in\Syn(\mathrsfs{C})$, which
is impossible.
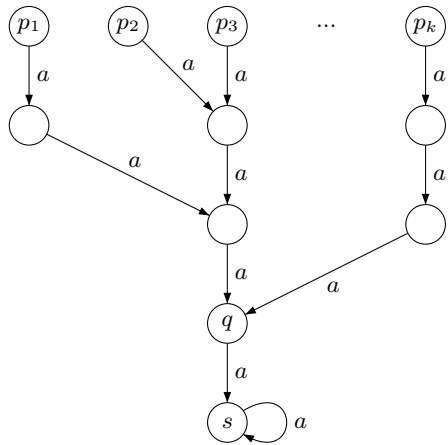
\begin{figure}[ht]
\begin{center}
\unitlength=2.5pt
  \begin{picture}(60,60)(0,0)
   \gasset{Nw=6,Nh=6,Nmr=3}
   \thinlines
   \node(p1)(0,60){$p_1$}
   \node(p2)(15,60){$p_2$}
   \node(p3)(30,60){$p_3$}
   \node(pk)(60,60){$p_k$}
   \node[Nframe=n](pn)(45,60){$...$}
   \node(s)(30,0){$s$}
   \node(q1)(0,45){}
   \node(q3)(30,45){}
   \node(qk)(60,45){}
   \node(t3)(30,30){}
   \node(tk)(60,30){}
   \node(v3)(30,15){$q$}
   \drawloop[loopdiam=6,loopangle=0](s){$a$}
   \drawedge(p1,q1){$a$}
   \drawedge(p2,q3){$a$}
   \drawedge(p3,q3){$a$}
   \drawedge(pk,qk){$a$}
   \drawedge(q3,t3){$a$}
   \drawedge(qk,tk){$a$}
   \drawedge(q1,t3){$a$}
   \drawedge(t3,v3){$a$}
   \drawedge(tk,v3){$a$}
   \drawedge(v3,s){$a$}
   \end{picture}
\end{center}
\caption{The action of $a$ in $\mathscr{C}$}
\label{Ca}
\end{figure}

Consider the set of leaves $H=Q\setminus{Q\dt a}=\{p_1,p_2,...,p_k\}$.
Since the DFA $\mathscr{C}$ is strongly connected, for each state $p_\ell$ in $H$ there exists a state $q_\ell$ such that $q_\ell\dt b=p_\ell$. Thus $H\subseteq Q\dt b$. We show that $H$ is exactly $Q\dt b$, meaning that  $Q\dt a\cap Q\dt b=\varnothing$. Take a leaf of height $n$. Without loss of generality suppose it is $p_1$. Let $q_1$ be such that $q_1\dt b=p_1$. The word $ba^{n-1}$ is synchronizing, so $Q\dt ba^{n-1}=\{q\}$ for some $q\in Q$. We have $q_1\dt ba^{n-1}=q$, and
 $q\dt a=s$ (see Fig.\ref{Ca}). Suppose there is $\overline{p}\in Q\dt a\cap Q\dt b.$ Then there is
a state $\overline{q}$ such that $\overline{q}\dt b=\overline{p}$. Since $\overline{p}$ is not a leaf,
we have $h(\overline{p})<n$. %Thus $\overline{p}\dt $Now we verify that $H = Q\dt b$. Arguing by contradiction assume that $H \neq Q\dt b$. It means that there exists a state $\overline{q}$ and a state $\overline{p}\in Q\dt b \setminus H$ with the property $\overline{q}\dt b=\overline{p}$. Note that $\overline{p}\dt a^l=s$ for some $l<n$.
Then $\overline{q}\dt ba^{n-1}=\overline{p}\dt a^{n-1}=s\neq q$. A contradiction. Hence $H = Q\dt b$. Furthermore, the height of any leaf of $\mathscr{C}_a$ is exactly $n$. To see this assume that there exists a state $p_m$ such that $h(p_m)<n$, i.e.\ $p_m\dt a^\ell=s$, for some $\ell<n$. Then the word $ba^{n-1}$ is not synchronizing. Indeed, take a state $q_m$ such that $q_m\dt b=p_m$. We have $q_m\dt ba^{n-1}=p_m\dt a^{n-1}=s\neq q$.

Consider an arbitrary state $p\in Q\dt a$. Let $\delta^{-1}(p,u)=\{p'\in Q\mid p'\dt u=p\}$. We prove that $|\delta^{-1}(p,a)|\geq 2$ for each $p\in Q\dt a$. For the root $s$ we have $\{s,q\}\subseteq \delta^{-1}(s,a)$, thus, $|\delta^{-1}(s,a)|\ge 2$.
%Note, that the state $q$ is the only state (except $s$) that maps to $s$ under the action of $a$. Indeed, let $\overline{q}\dt a=s$ for some $\overline{q}\neq q$. Find a leaf $p_m$ such that $p_m\dt a^{n-1}=\overline{q}$. Then $\{q,\overline{q}\} \subseteq Q\dt ba^{n-1}$. But  $Q\dt ba^{n-1}=q$ by the definition of $q$. Contradiction. So $\mid \delta^{-1}(s,a)\mid=2$.
%Since the word $ba^{n-2}$ does not synchronize $\mathscr{C}$, we have $|Q\dt ba^{n-2}|\geq 2$. But $(Q\dt ba^{n-2})\dt a=\{q\}$, therefore $|\delta^{-1}(q,a)|\geq 2$.
Let $p$ be an arbitrary state in $Q\dt a$. Strong connectivity of $\mathrsfs{C}$ implies that
there exists a state $\overline{p}$ and a word $w\in \Sigma^n$ such that $\overline{p}\dt w=p$. Since $w$ is synchronizing, we have $Q\dt w=\{p\}$. Consider the word $w[1..n-1]$ that does not synchronize $\mathrsfs{C}$. Then $|Q\dt w[1..n-1]|\geq 2$. However, $(Q\dt w[1..n-1])\dt w[n]=p$. And we obtain the inequality $|\delta^{-1}(p,a)|\geq 2$.
Denote $H_0=\{q\}$ and construct sets $H_{i}=\delta^{-1}(H_{i-1},a)$ for $1\leq i \leq n-1$. We have $|H_i|\geq 2^i$ for all $1\leq i\leq n-1$. Then $\mathscr{C}$ possesses at least $1+1+2+4+...+2^{n-1}=2^n$ states.

Thus we have $|Q|=2^n$. Moreover, $Q=\cup_{|w|=n}Q\dt w$. It means that with each state $q$ of $Q$ we can associate the word $w$
of length $n$ such that $Q\dt w=\{q\}$. It is clear that it gives us the desired isomorphism between $\mathrsfs{C}$ and $\mathrsfs{B}$.

\qed

\end{proof}

\begin{remark} In case $\Sigma=\{a,b\}\cup\Delta$, where $\Delta\ne\varnothing$, we consider De Bruijn automaton constructed for the binary alphabet $\{a,b\}$
and put the action of each letter in $\Delta$ to be the same as the action of the letter $a$. It is clear that the language of synchronizing words
of the modified De Bruijn automaton coincides with $\Sigma^{\ge n}$.
\end{remark}

The Proposition implies that the minimal DFA recognizing an ideal language $L$ can be exponentially smaller than a strongly connected MSA $\mathscr{B}$ with $Syn(\mathscr{B})=L$.

\subsection{Ideal language generated by a set of words of fixed length}

\begin{theorem}
Let $U\subsetneq \Sigma^n$. There is a strongly connected synchronizing automaton $\mathrsfs{B}_U$ with $2^n$ states such that
$\Syn(\mathrsfs{B}_U)=\Sigma^*U\Sigma^*$.
\end{theorem}

\begin{proof}We modify the De Bruijn automaton $\mathrsfs{B}$ from the section 2.1 to
obtain the desired automaton $\mathrsfs{B}_U$. First of all it is convenient to
view the states of the automaton $\mathrsfs{B}$ not as the words of length $n$,
but as pairs $(x,u)$, where $x\in\Sigma$ and $u\in\Sigma^{n-1}$. Then by the
definition of the transitions in $\mathrsfs{B}$ we have%there is a transition labeled $y\in\Sigma$
%from the state $(x,u)$ to the state $(z,v)$ if and only if $uy=zv$
\begin{equation}\label{B_trans}(x,u)\xrightarrow{y}(z,v) \Leftrightarrow uy=zv\end{equation}
For a word $uy$ which is not in $U$, we modify the corresponding transition
given by \eqref{B_trans} in the following way.
%For a state of the form $(x,u)$ and any letter $y\in\Sigma$ such that
%the word $uy$ is synchronizing, the transition remains unchanged,
%i.e.\ $(x,u)\dt y=(z,v)$, where $z\in\Sigma$ and $v\in\Sigma^{n-1}$
%are such that $uy=zv$.
%Let $u\notin U\cup\{a^n,b^n\}$, $x\in\Sigma$.
If $uy\notin U\cup\{a^n,b^n\}$ we put \begin{equation}\label{B_U_trans}(x,u)\xrightarrow{y}(x,v),\end{equation} where $v$ is
defined by \eqref{B_trans}.

If $uy=a^n\notin U$ ($uy=b^n\notin U$ respectively) we put
\begin{equation}\label{B_U_ab_trans}(a,a^{n-1})\xrightarrow{a}(b,a^{n-1}),\quad ((b,b^{n-1})\xrightarrow{b}(a,b^{n-1}) \text{ respectively}).\end{equation}
%If $uy=a^n\notin U$ we put
%\begin{equation}\label{B_U_a_trans}(a,a^{n-1})\xrightarrow{a}(b,a^{n-1}) \text{ and } (b,a^{n-1})\xrightarrow{a}(a,a^{n-1}).\end{equation}
%
%If $uy=b^n\notin U$ we put
%\begin{equation}\label{B_U_b_trans}(b,b^{n-1})\xrightarrow{b}(a,b^{n-1}) \text{ and } (a,b^{n-1})\xrightarrow{b}(b,b^{n-1}).\end{equation}
\begin{figure}[ht]
\begin{center}
  \begin{picture}(60,48)
   \gasset{Nadjust=w,Nadjustdist=1.5,Nh=6,Nmr=3}
   \thinlines
   \node[Nw=9](aaa)(0,35){$(a,aa)$}
   \node[Nw=9](aab)(20,50){$(a,ab)$}
   \node[Nw=9](aba)(20,35){$(a,ba)$}
   \node[Nw=9](abb)(40,50){$(a,bb)$}
   \node[Nw=9](baa)(0,0){$(b,aa)$}
   \node[Nw=9](bab)(20,20){$(b,ab)$}
   \node[Nw=9](bba)(20,0){$(b,ba)$}
   \node[Nw=9](bbb)(40,20){$(b,bb)$}
   \drawloop[loopdiam=6,loopangle=90](aaa){$a$}
   \drawloop[loopdiam=6,loopangle=0](bbb){$b$}
   \drawedge(aaa,aab){$b$}
   \drawedge(aab,aba){$a$}
   \drawedge(aab,abb){$b$}
   \drawedge[curvedepth=-2,sxo=-2](aba,baa){$a$}
   \drawedge(aba,bab){$b$}
   \drawedge[curvedepth=5,sxo=-4](baa,aaa){$a$}
   \drawedge[curvedepth=3,sxo=-2](baa,aab){$b$}
  \drawedge[sxo=1](abb,bba){$a$}
   \drawedge[curvedepth=5,sxo=4](abb,bbb){$b$}
   \drawedge[curvedepth=3](bab,aba){$a$}
   \drawedge(bab,abb){$b$}
   \drawedge(bbb,bba){$a$}
   \drawedge(bba,baa){$a$}
   \drawedge(bba,bab){$b$}
   \end{picture}
\end{center}
\caption{De Bruijn automaton for $n=3$}
\label{DeBr3_1}
\end{figure}
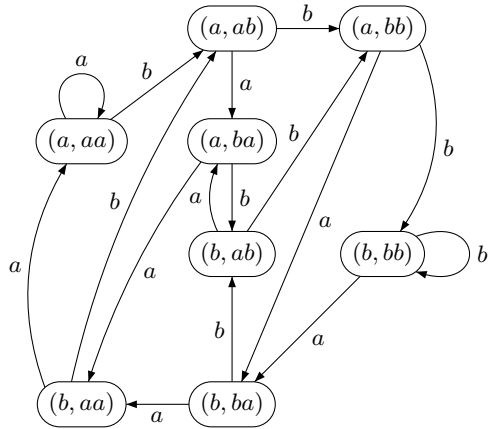
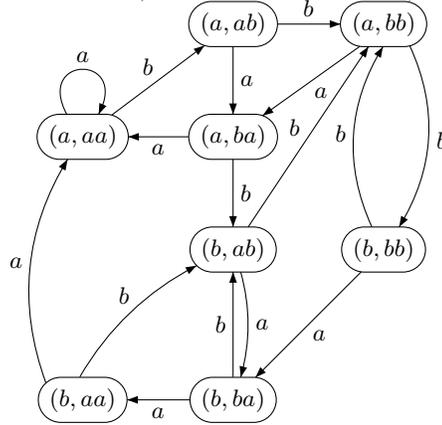
\begin{figure}[ht]
\begin{center}
  \begin{picture}(60,48)
   \gasset{Nadjust=w,Nadjustdist=1.5,Nh=6,Nmr=3}
   \thinlines
   \node[Nw=9](aaa)(0,35){$(a,aa)$}
   \node[Nw=9](aab)(20,50){$(a,ab)$}
   \node[Nw=9](aba)(20,35){$(a,ba)$}
   \node[Nw=9](abb)(40,50){$(a,bb)$}
   \node[Nw=9](baa)(0,0){$(b,aa)$}
   \node[Nw=9](bab)(20,20){$(b,ab)$}
   \node[Nw=9](bba)(20,0){$(b,ba)$}
   \node[Nw=9](bbb)(40,20){$(b,bb)$}
   \drawloop[loopdiam=6,loopangle=90](aaa){$a$}
%   \drawloop[loopdiam=6,loopangle=0](bbb){$b$}
   \drawedge(aaa,aab){$b$}
   \drawedge(aab,aba){$a$}
   \drawedge(aab,abb){$b$}
   \drawedge(aba,aaa){$a$}
   \drawedge(aba,bab){$b$}
   \drawedge[curvedepth=5,sxo=-4](baa,aaa){$a$}
   \drawedge[curvedepth=3,sxo=-2](baa,bab){$b$}
  \drawedge[sxo=1](abb,aba){$a$}
   \drawedge[curvedepth=5,sxo=2](abb,bbb){$b$}
  \drawedge[curvedepth=5,exo=2](bbb,abb){$b$}
   \drawedge[curvedepth=2](bab,bba){$a$}
   \drawedge(bab,abb){$b$}
   \drawedge(bbb,bba){$a$}
   \drawedge(bba,baa){$a$}
   \drawedge(bba,bab){$b$}
   \end{picture}
\end{center}
\caption{Automaton $\mathrsfs{B}_U$ for $U=\{aaa, abb, bab\}$}
\label{DeBr3_2}
\end{figure}
The other transitions remain unchanged. The obtained automaton is denoted by $\mathrsfs{B}_U$.
The examples of the automaton $\mathrsfs{B}$ and the corresponding modified automaton $\mathrsfs{B}_U$ for
$U=\{aaa,abb,bab\}$ are shown on Fig.\ref{DeBr3_1} and Fig.\ref{DeBr3_2} respectively.
We prove that the automaton $\mathrsfs{B}_U$ satisfies the statement of the proposition.
First we show that $\mathrsfs{B}_U$ is strongly connected.
For this purpose we prove that all the states are reachable from the state $(a,a^{n-1})$,
and the state $(a,a^{n-1})$ is reachable from all states.

First we show that a state $(a,u)$ is reachable from $(a,a^{n-1})$ for any $u\in\Sigma^{n-1}$.
If  $u=a^{n-1}$, the claim obviously holds. Hence we may assume $u=a^kb\hat{u}$, where $k\ge0$,
$\hat{u}\in\Sigma^{n-k-2}$. By the definition of transitions in $\mathrsfs{B}_U$ we have
$$(a,a^{n-1})\xrightarrow{b}(a,a^{n-2}b)\xrightarrow{\hat{u}[1]}(a,a^{n-3}b\hat{u}[1])\xrightarrow{\hat{u}[2]}\cdots\xrightarrow{\hat{u}[n-k-3]}$$
$$(a,a^{k+1}b\hat{u}[1..n-k-3])\xrightarrow{\hat{u}[n-k-2]}(a,a^kb\hat{u}[1..n-k-2])=(a,u).$$
Symmetrically any state $(b,u)$ is reachable from the state $(b,b^{n-1})$. The latter state is reachable
from $(a,b^{n-1})$. Thus the state $(b,u)$ is reachable also from $(a,a^{n-1})$:
$$(a,a^{n-1})\rightsquigarrow (a,b^{n-1})\xrightarrow{b}(b,b^{n-1})\rightsquigarrow (b,u).$$

Now we show that the state $(a,a^{n-1})$ is reachable from any other state. Apply the word $a^{n-1}$ to an
arbitrary state $(x,u)$. By the definition of transitions we have $(x,u)\dt a^{n-1}\in \{(a,a^{n-1}),(b,a^{n-1})\}.$
If $(x,u)\dt a^{n-1}=(a,a^{n-1})$ we are done. If $(x,u)\dt a^{n-1}=(b,a^{n-1})$, then we apply once more the letter $a$ and obtain
$(x,u)\dt a^{n}=(a,a^{n-1})$.

Thus the constructed automaton $\mathrsfs{B}_U$ is strongly connected.
Next we show that $\Syn(\mathrsfs{B}_U)=\Sigma^*U\Sigma^*.$ It is easy to see that for any word $u\in\Sigma^{n-1}$
we have $Q\dt u\subseteq\{(a,u),(b,u)\}$, and $Q\dt u\cap Q\dt v=\varnothing$ for $u,v\in\Sigma^{n-1}$ such that $u\ne v$. Thus $Q\supseteq\bigcup\limits_{|u|=n-1}Q\dt u$.
Next we check that $Q=\bigcup\limits_{|u|=n-1}Q\dt u$. Indeed, if $a^n\in U$ we have $(a,a^{n-1})\xrightarrow{u}(a,u)$ for all $u\in\Sigma^{n-1}$. If $a^n\not\in U$ take any word $u\in\Sigma^{n-1}$. If $u=a^{n-1}$ then $u$ maps the state $(a,a^{n-1})$ or the state $(b,a^{n-1})$ to $(a,u)$. Let us assume now that $u=a^kb\hat{u}$. If $k$ is even (odd, respectively) then $u$ maps $(a,a^{n-1})$ ($(b,a^{n-1})$, respectively) to $(a,u)$. So any states $(a,u)$ belongs to the set $\bigcup\limits_{|u|=n-1}Q\dt u$. Symmetrically any states $(b,u)$ belongs to the latter set. Hence $Q=\bigcup\limits_{|u|=n-1}Q\dt u$.
Since $|Q|=2^n$, if there is a synchronizing word $u$ of length $n-1$, we would have
$2^n=|Q|=|\bigcup\limits_{|u|=n-1}Q\dt u|<2^n$, which is a contradiction. Thus, none of the words
of length $n-1$ is synchronizing. Consider an arbitrary word $w$ of length $n$
and factorize it as $w=uy$ with $u\in\Sigma^{n-1}$ and $y\in\Sigma$. We have $Q\dt u=\{(a,u),(b,u)\}.$
If $w\in U$, then the corresponding transitions from the states $(a,u)$ and $(b,u)$
were not changed, and we have $Q\dt uy=\{(z,v)\}$, where $uy=zv$, so $w$ is synchronizing.
If $w\notin U$, then $Q\dt uy=\{(a,v),(b,v)\}$, where $v$ is such that $uy=zv$ for some $z\in\Sigma$,
so $w\notin\Syn(\mathrsfs{B}_U)$.

\qed
\end{proof}

\subsection{Ideal languages generated by a finite set of words}

\begin{theorem}
Let $S$ be finite and anti-factorial set of words in $\Sigma^+$. There is a strongly
connected synchronizing automaton $\mathrsfs{C}_S$ such that $\Syn(\mathrsfs{C}_S)=\Sigma^*S\Sigma^*.$
This automaton has at most $2^n$ states, where $n=\max{\{|s|\mid s\in S\}}$.
\end{theorem}

\begin{proof} Let $T=\{w\in\Sigma^n\mid \exists s\in S, s\in \Fact(w)\}.$ First we construct the automaton
$\mathrsfs{B}_{T}$ as described in the previous proposition.
In that proposition the states of $\mathrsfs{B}_{T}$ were viewed as pairs $(x,u)$ with $x\in\Sigma$, $u\in\Sigma^{n-1}$.
Here it will be convenient to view the states as the words $xu$ of length $n$
(as it was in the initial De Bruijn automaton). Note, that since $S$ is anti-factorial,
every state in $T$ can be uniquely
factorized as $usv$ such that $s\in S$, $u,v\in\Sigma^*$ and $sv$ does not contain
factors in $S$ except $s$. In what follows we will use this unique representation without stating it explicitly.

Next we define an equivalence relation $\simeq$ on the set
of states of this automaton (i.e.\ on words of length $n$) in the following way.
Let $w,w'\in T$. We have $w\simeq w'$ iff $w=usv$ and $w'=u'sv$, where $s\in S$,
$u,u',v\in\Sigma^*$. On the set $\Sigma^n\setminus T$ the relation $\simeq$ is defined trivially,
i.e.\ for $w,w'\in \Sigma^n\setminus T$ we have $w\simeq w'$ iff $w=w'$.
It is easy to see that $\simeq$ is indeed an equivalence relation on $\Sigma^n$. In fact,
$\simeq$ is a congruence on the set of states of the automaton $\mathrsfs{B}_T$.
Let us check that for any $x\in\Sigma$ and any $w,w'\in\Sigma^n$ $w\simeq w'$ implies
$w\dt x \simeq w'\dt x$. If $w,w'\in \Sigma^n\setminus T$, then $w=w'$ and we are done.
if $w,w'\in T$, then $w=usv$, $w'=u'sv$. If $u=u'=\varepsilon$, then
$w=w'$, and there is nothing to prove. So we may assume, that $u,u'\ne\varepsilon$. Then
$usv\dt x=tsvx$ and $u'sv\dt x= t'svx$ for some $t,t'\in\Sigma^*$. Since the obtained two words
have the same suffixes, containing a word in $S$, they are equivalent.
So we can consider the factor automaton $\mathrsfs{B}_T/\simeq$, whose states are the equivalence
classes of $\simeq$, and the transition function is induced from the initial automaton.
Let us denote by $[sv]$ the equivalence class of a word $usv\in T$, and by $[u]$ the equivalence class
of a word $u\notin T$.
We claim, that $\mathrsfs{C}_S=\mathrsfs{B}_T/\simeq$. In other words, the constructed automaton is strongly connected,
and $\Syn(\mathrsfs{B}_T/\simeq)=\Sigma^*S\Sigma^*$. The first property holds trivially, since a factor automaton
of a strongly connected automaton is strongly connected.

For any $w\in\Sigma^*$ and $s\in S$ previously in $\mathrsfs{B}_T$ we had $w\dt s=us$, where $u\in\Sigma^*$. Since $S$
is anti-factorial, in $\mathrsfs{B}_T/\simeq$ we have $[w]\dt s=[s]$, so any $s\in S$ is synchronizing for
the automaton $\mathrsfs{B}_T/\simeq$. Now let $t$ be a synchronizing word, so there is a state $[w]$
such that for any state $[w']$ we have $[w']\dt t=[w]$. If $[w]$ is a one-element class, then the word
$t$ was synchronizing for the initial automaton $\mathrsfs{B}_T$, so $t$ contains some word in $S$ as a factor,
i.e.\ $t\in \Sigma^*S\Sigma^*$. Consider the case where $[w]$ is a class consisting of elements $u_1sv,\ u_2sv,\ldots, u_ksv$, $k>1$.
Note that in this case $u_i\ne\varepsilon$ for each $i=1,\ldots,k$. This means that $t=usv$ for some $u\in\Sigma^*$,
thus, also in this case $t\in \Sigma^*S\Sigma^*$.

\qed

\end{proof}

Complete this section with an example. Let $S=\{a^2,aba\}$ and $\Sigma=\{a,b\}$. Construct the corresponding set $T=\{a^3,a^2b,ba^2,aba\}$. Next build the DFA $\mathrsfs{B}_T$. The resulting automaton is shown on the left side of Fig.\ref{DeBrTT}.
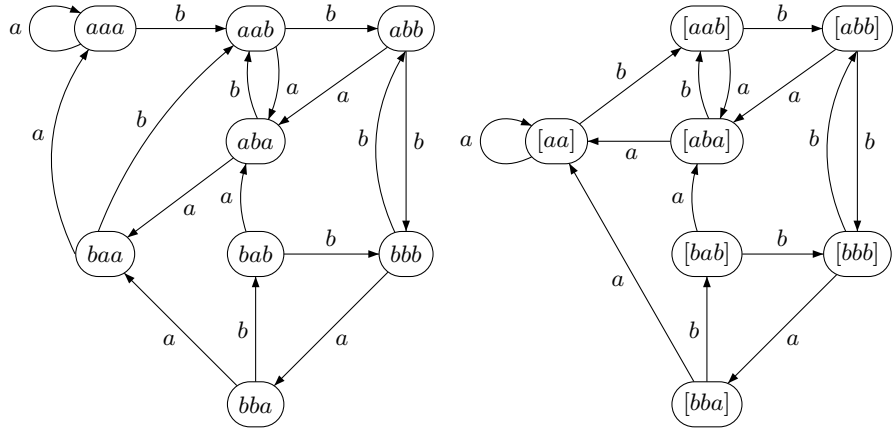
\begin{figure}[ht]
\begin{center}
  \begin{picture}(105,48)(0,3)
   \gasset{Nadjust=w,Nadjustdist=1.5,Nh=6,Nmr=3}
   \thinlines
   \node[Nw=9](aaa)(5,50){$aaa$}
   \node[Nw=9](aab)(25,50){$aab$}
   \node[Nw=9](aba)(25,35){$aba$}
   \node[Nw=9](abb)(45,50){$abb$}
   \node[Nw=9](baa)(5,20){$baa$}
   \node[Nw=9](bab)(25,20){$bab$}
   \node[Nw=9](bba)(25,0){$bba$}
   \node[Nw=9](bbb)(45,20){$bbb$}
   \drawloop[loopdiam=6,loopangle=180](aaa){$a$}
%   \drawloop[loopdiam=6,loopangle=0](bbb){$b$}
   \drawedge(aaa,aab){$b$}
   \drawedge[curvedepth=2,sxo=2](aab,aba){$a$}
   \drawedge(aab,abb){$b$}
   \drawedge[curvedepth=2,sxo=2](aba,aab){$b$}
   \drawedge(aba,baa){$a$}
   \drawedge[curvedepth=5,sxo=-4](baa,aaa){$a$}
   \drawedge[curvedepth=3,sxo=-2](baa,aab){$b$}
  \drawedge[sxo=1](abb,aba){$a$}
   \drawedge(abb,bbb){$b$}
  \drawedge[curvedepth=5,exo=2](bbb,abb){$b$}
   \drawedge[curvedepth=2](bab,aba){$a$}
   \drawedge(bab,bbb){$b$}
   \drawedge(bbb,bba){$a$}
   \drawedge(bba,baa){$a$}
   \drawedge(bba,bab){$b$}
      \node[Nw=9](aaa1)(65,35){$[aa]$}
   \node[Nw=9](aab1)(85,50){$[aab]$}
   \node[Nw=9](aba1)(85,35){$[aba]$}
   \node[Nw=9](abb1)(105,50){$[abb]$}
   %\node[Nw=9](baa1)(60,20){$baa$}
   \node[Nw=9](bab1)(85,20){$[bab]$}
   \node[Nw=9](bba1)(85,0){$[bba]$}
   \node[Nw=9](bbb1)(105,20){$[bbb]$}
   \drawloop[loopdiam=6,loopangle=180](aaa1){$a$}
%   \drawloop[loopdiam=6,loopangle=0](bbb1){$b$}
   \drawedge(aaa1,aab1){$b$}
   \drawedge[curvedepth=2,sxo=2](aab1,aba1){$a$}
   \drawedge(aab1,abb1){$b$}
   \drawedge[curvedepth=2,sxo=2](aba1,aab1){$b$}
   \drawedge(aba1,aaa1){$a$}
   %\drawedge[curvedepth=5,sxo=-4](baa1,aaa1){$a$}
   %\drawedge[curvedepth=3,sxo=-2](baa1,aab1){$b$}
  \drawedge[sxo=1](abb1,aba1){$a$}
   \drawedge(abb1,bbb1){$b$}
  \drawedge[curvedepth=5,exo=2](bbb1,abb1){$b$}
   \drawedge[curvedepth=2](bab1,aba1){$a$}
   \drawedge(bab1,bbb1){$b$}
   \drawedge(bbb1,bba1){$a$}
   \drawedge(bba1,aaa1){$a$}
   \drawedge(bba1,bab1){$b$}
   \end{picture}
\end{center}
\caption{Automata $\mathrsfs{B}_T$ and $\mathrsfs{B}_T/\simeq$ for $T=\{aaa, aab, baa,aba\}$}
\label{DeBrTT}
\end{figure}

By the definition of $\simeq$ the class  $[aa]$ includes states $aaa$ and $baa$. The rest classes are one-element. The resulting automaton $\mathrsfs{B}_T/\simeq$ is shown on the right side of Fig.\ref{DeBrTT}

\subsection{Ideal languages generated by two words}

%Let $w_1,w_2\in\Sigma^*$ and $|\Sigma|>1$. Denote $|w_1|=n$ and $|w_2|=m$. Assume that $w_1\not\in \{ab^{n-1},a^{n-1}b\}$ and $w_2\not\in \{ab^{m-1},a^{m-1}b\}$ for $a,b\in \Sigma$ and $a\neq b$. Let $S=\{w_1,w_2\}$ and $S$ is anti-factorial. Now we prove that $rcs(L)\leq n+m$ and the bound is tight.

Let $S=\{u,v\}\subseteq\Sigma^+$ and let $|u|=n$, $|v|=m$. Again we suppose that $S$ is anti-factorial. In this case we can construct a strongly connected automaton
$\mathrsfs{D}_{u,v}$ such that $\Syn(\mathrsfs{D}_{u,v})=\Sigma^*(u+v)\Sigma^*$ with $n+m$ states, thus improving construction
from the previous section.
%Let $L_{u,v}=\Sigma^*(u+v)\Sigma^*$. We assume that neither $u$ is a factor of $v$, nor $v$
% is a factor of $u$. We prove that in most cases $rcs(L_{u,v})\leq n+m$. Moreover, we show that this upper bound cannot be improved by
% exhibiting two words $u$ and $v$ such that $\rcs(L_{u,v})=|u|+|v|.$
%Assume that $u\not\in \{ab^{n-1},a^{n-1}b\}$ and $v\not\in \{ab^{m-1},a^{m-1}b\}$ for different letters $a, b\in \Sigma$.
For simplicity we state and prove the following theorem only for the case of binary alphabet, although the same argument works in general.
\begin{theorem}
Let $\Sigma=\{a,b\}$, and let $u\in\Sigma^n\setminus\{ab^{n-1},a^{n-1}b,ba^{n-1},b^{n-1}a\}$, $v\in\Sigma^m\setminus\{ab^{m-1},a^{m-1}b,ba^{m-1},b^{m-1}a\}$.
There is a strongly connected synchronizing automaton $\mathrsfs{D}_{u,v}$ having $n+m$ states such that $\Syn(\mathrsfs{D}_{u,v})=\Sigma^*(u+v)\Sigma^*$.
\end{theorem}

\begin{proof}
In order to obtain $\mathrsfs{D}_{u,v}$ we combine minimal automata for the languages $\Sigma^*u\Sigma^*$ and $\Sigma^*v\Sigma^*$.
For a letter $x\in\{a,b\}$ by $\overline{x}$ we denote its \emph{complementary} letter, i.e. $\overline{a}=b$, and $\overline{b}=a$.
Recall the construction of the minimal automaton recognizing the language $\Sigma^{*}w\Sigma^{*}$, where $w\in\Sigma^+$.
%Let us remind the construction of the minimal DFA recognizing a principal ideal language, i.e. an ideal language generated by a single word.
%Fix a word $w$ over $\Sigma=\{a,b\}$. Let $|w|=k$.
It is well-known that this automaton has $|w|+1$ states.
We enumerate the states of this automaton by the prefixes of the word $w$ so that the state $w[1..i]$ maps to the state $w[1..i+1]$ under the action of the letter $w[i+1]$ for all $i$, $0\leq i<k$. The other letter $\overline{w[i+1]}$ sends the state $w[1..i]$ to state $p$ such that $p$ is the maximal prefix of $w$ that appears in $w[1..i+1]$ as a suffix. The state $w$ is the sink state of the automaton.
The initial state is $\varepsilon$ and the unique final state is $w$, see Fig.\ref{min_aut} (the transitions labeled by complementary letters $\overline{w[i]}$ are not shown).

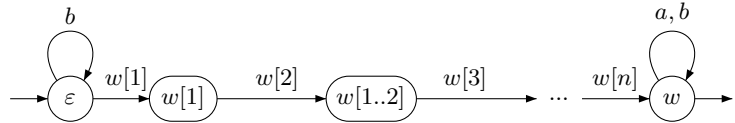
\begin{figure}[ht]
\begin{center}
  \begin{picture}(95,10)
   \gasset{Nw=6,Nh=6,Nmr=3}
   \thinlines
   \node[Nmarks=i](A1)(10,0){$\varepsilon$}
   \node[Nw=9](A2)(25,0){$w[1]$}
   \node[Nw=12](A3)(50,0){$w[1..2]$}
   \node[Nmarks=f](A4)(90,0){$w$}
   \node[Nframe=n](A5)(75,0){$...$}
%   \node[Nframe=n](Ai)(0,0){$$}
%   \node[Nframe=n](At)(65,0){$$}
   \drawloop[loopdiam=6,loopangle=90](A4){$a,b$}
   \drawloop[loopdiam=6,loopangle=90](A1){$b$}
   \drawedge(A1,A2){$w[1]$}
   \drawedge(A2,A3){$w[2]$}
   \drawedge(A5,A4){$w[n]$}
   \drawedge(A3,A5){$w[3]$}
%   \drawedge(A3,At){$$}
   \end{picture}
\end{center}
\caption{The minimal DFA $\mathscr{A}_w$.}
\label{min_aut}
\end{figure}

Construct minimal automata $\mathscr{A}_{u}$ and $\mathscr{A}_{v}$. Denote by $\mathrsfs{A}'_{u}$ the automaton obtained from $\mathrsfs{A}_{u}$ by deleting the sink state and the transition from $u[1..n-1]$ labeled by $u[n]$. Denote by $\mathrsfs{A}'_{v}$ the corresponding automaton for $v$. Define the action of letters $u[n]$ and $v[m]$ on states $u[1..n-1]$ and $v[1..m-1]$ as follows. Denote by $p$ the state in $\mathrsfs{A}'_u$ corresponding to the maximal prefix of $u$ that appears in $v$ as a suffix. Denote by $s$ the state in $\mathrsfs{A}'_v$ corresponding to the maximal prefix of $v$ that appears in $u$ as a suffix. %For example, take words $u=aaabaa$ and $v=baabba$, then $p=a$ and $s=baa$.
We put $u[1..n-1]\dt u[n]=s$ and $v[1..m-1]\dt v[m]=p$.
%If $s$ is empty word then $u[n]$ maps $u[1..n-1]$ to the state $\varepsilon$ in $\mathrsfs{A}'_{v}$. The case when $p$ is empty word is considered analogously. Assume that $p\neq \varepsilon$ and $s\neq \varepsilon$. Then $u[n]$ maps $u[1..n-1]$ to $v[1..|s|]$ and $v[m]$ maps $v[1..m-1]$ to $u[1..|p|]$.
Denote the resulting automaton by $\mathrsfs{D}_{u,v}$ and prove that it satisfies the desired properties.
Figures \ref{ExTwoWords1},\ref{ExTwoWords2},\ref{ExTwoWordsB} illustrate the construction for $u=abaab$ and $v=babab$.

\begin{figure}[ht]
\begin{center}
\begin{picture}(75,12)(0,2)
   \gasset{Nw=6,Nh=6,Nmr=3}
\thinlines
\node(A)(0,8){$\varepsilon$}
\node(B)(13,8){$a$}
\node[Nw=7](C)(27,8){$ab$}
\node[Nw=8](D)(43,8){$aba$}
\node[Nw=9](E)(59,8){$abaa$}
\node[Nw=10](F)(75,8){$abaab$}
\drawedge(A,B){$a$}
\drawedge(B,C){$b$}
\drawedge(C,D){$a$}
\drawedge(D,E){$a$}
\drawedge(E,F){$b$}
\drawloop[loopdiam=6,loopangle=90](A){$b$}
\drawloop[loopdiam=6,loopangle=90](B){$a$}
\drawloop[loopdiam=6,loopangle=90](F){$a,b$}
\drawedge[curvedepth=5](C,A){$b$}
\drawedge[curvedepth=-8,ELside=r](E,B){$a$}
\drawedge[curvedepth=-3,ELside=r](D,C){$b$}
\end{picture}
\end{center}
\caption{The minimal DFA recognizing $\Sigma^{*}abaab\Sigma^{*}$.}
\label{ExTwoWords1}
\end{figure}
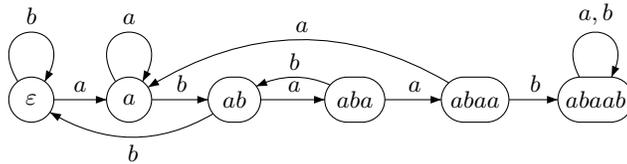
\begin{figure}[ht]
\begin{center}
\begin{picture}(75,12)
   \gasset{Nw=6,Nh=6,Nmr=3}
\thinlines
\node(A)(0,8){$\varepsilon$}
\node(B)(13,8){$b$}
\node[Nw=7](C)(27,8){$ba$}
\node[Nw=8](D)(43,8){$bab$}
\node[Nw=9](E)(59,8){$baba$}
\node[Nw=10](F)(75,8){$babab$}
\drawedge(A,B){$b$}
\drawedge(B,C){$a$}
\drawedge(C,D){$b$}
\drawedge(D,E){$a$}
\drawedge(E,F){$b$}
\drawloop[loopdiam=6,loopangle=90](A){$a$}
\drawloop[loopdiam=6,loopangle=90](B){$b$}
\drawloop[loopdiam=6,loopangle=90](F){$a,b$}
\drawedge[curvedepth=5](C,A){$a$}
\drawedge[curvedepth=-7,ELside=r](D,B){$b$}
\drawedge[curvedepth=12](E,A){$a$}
\end{picture}
\end{center}
\caption{The minimal DFA recognizing $\Sigma^{*}babab\Sigma^{*}$.}
\label{ExTwoWords2}
\end{figure}
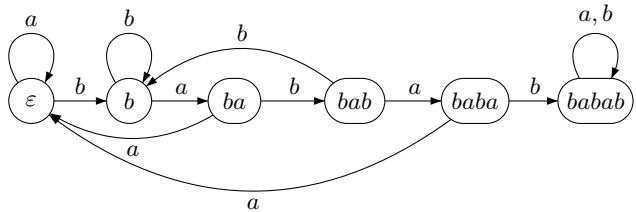
\begin{figure}[ht]
\begin{center}
\begin{picture}(60,45)
   \gasset{Nw=6,Nh=6,Nmr=3}
\thinlines
\node(A1)(0,8){$\varepsilon$}
\node(B1)(13,8){$b$}
\node[Nw=7](C1)(27,8){$ba$}
\node[Nw=8](D1)(43,8){$bab$}
\node[Nw=9](E1)(59,8){$baba$}
%\node[Nw=10](F)(75,7){$babab$}
\drawedge(A1,B1){$b$}
\drawedge(B1,C1){$a$}
\drawedge(C1,D1){$b$}
\drawedge(D1,E1){$a$}
%\drawedge(E,F){$b$}
\drawloop[loopdiam=6,loopangle=90](A1){$a$}
\drawloop[loopdiam=6,loopangle=90](B1){$b$}
%\drawloop[loopdiam=6,loopangle=90](F){$a,b$}
\drawedge[curvedepth=5](C1,A1){$a$}
\drawedge[curvedepth=-6,ELside=r](D1,B1){$b$}
\drawedge[curvedepth=12](E1,A1){$a$}
\node(A2)(0,36){$\varepsilon$}
\node(B2)(13,36){$a$}
\node[Nw=7](C2)(27,36){$ab$}
\node[Nw=8](D2)(43,36){$aba$}
\node[Nw=9](E2)(59,36){$abaa$}
%\node[Nw=10](F)(75,7){$abaab$}
\drawedge(A2,B2){$a$}
\drawedge(B2,C2){$b$}
\drawedge(C2,D2){$a$}
\drawedge(D2,E2){$a$}
%\drawedge(E,F){$b$}
\drawloop[loopdiam=6,loopangle=90](A2){$b$}
\drawloop[loopdiam=6,loopangle=90](B2){$a$}
%\drawloop[loopdiam=6,loopangle=90](F){$a,b$}
\drawedge[curvedepth=5](C2,A2){$b$}
\drawedge[curvedepth=-8,ELside=r](E2,B2){$a$}
\drawedge[curvedepth=-3,ELside=r](D2,C2){$b$}
\drawedge[curvedepth=-3,ELside=r,dash={1.5}0](E2,B1){$b$}
\drawedge[ELside=r,dash={1.5}0](E1,C2){$b$}
\end{picture}
\end{center}
\caption{The DFA $\mathrsfs{D}_{u,v}$.}
\label{ExTwoWordsB}
\end{figure}
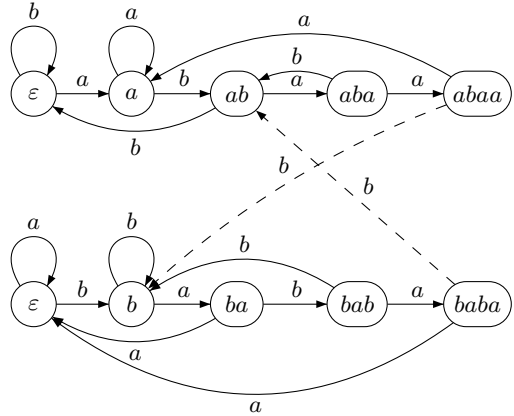
\noindent The following claim is rather easy to see. The explicit proof can be found in \cite{PrincIdArxiv}.
\begin{claim}
If $w\in\Sigma^n\setminus\{a^{n-1}b, ab^{n-1}\}$, then the automaton $\mathrsfs{A}'_w$ is strongly connected.
\end{claim}
%\begin{proof}
%Consider an arbitrary state $w[1..i]$
% of $\mathrsfs{A}'_w$. This state is obviously reachable from the state $\varepsilon$. We show that the state $\varepsilon$ is reachable from $w[1..i]$.
% Let $w[1..j]=w[1..i]\dt c^{|w|}$, where $c=\overline{w[i+1]}$. The state $w[1..j]$ is a maximal prefix of $w$ of the form $c^k$. If $c=b$,
% then $w[1..j]=\varepsilon$, so we are done. Suppose $c=a$. Apply $b$ to the state $w[1..j]$ (this is possible, since $w\ne a^{n-1}b$). Next we apply $\overline{w[j+2]}$. If $\overline{w[j+2]}=b$, then $w[1..j]\dt bb =\varepsilon$. If $\overline{w[j+2]}=a$, then
% $w[1..j]\dt bab =\varepsilon$ (since $w\ne ab^{n-1}$). So in both cases the state $\varepsilon$ is reachable from $w[1..i]$.
% \end{proof}
% \qed
By the Claim automata $\mathrsfs{A}'_{u}$ and $\mathrsfs{A}'_{v}$ are strongly connected. By the definition of the action of letters $u[n]$ and $v[m]$ on states $u[1..n-1]$ and $v[1..m-1]$, the resulting automaton $\mathrsfs{D}_{u,v}$ is also strongly connected. %, see Fig.~\ref{Structure}.
%
%\begin{figure}[ht]
%\begin{center}
%\begin{picture}(25,12)
%   \gasset{Nw=6,Nh=7,Nmr=3}
%\thinlines
%\node[Nw=10,dash={1.5}0](A)(0,6){$\mathrsfs{A}'_{u}$}
%\node[Nw=10,dash={1.5}0](B)(25,6){$\mathrsfs{A}'_{v}$}
%\drawedge[curvedepth=3](A,B){$u[n]$}
%\drawedge[curvedepth=3](B,A){$v[m]$}
%\end{picture}
%\end{center}
%\caption{Structure of the DFA $\mathscr{B}$.}
%\label{Structure}
%\end{figure}

 Now we are going to verify that  $u,v\in \Syn(\mathrsfs{D}_{u,v})$. 
 The state set of $\mathrsfs{D}_{u,v}$ is the union of the state set of $\mathrsfs{A}'_{u}$ (denoted by $Q^u$), and the state set of $\mathrsfs{A}'_{v}$ (denoted by $Q^v$). To avoid confusion when necessary we will use the upper indices $u$ and $v$ for the states in $Q^u$ and $Q^v$ respectively.
Let $w$ be an arbitrary word. We claim that $\varepsilon^{u} \dt w = r$, where $r$ is the maximal prefix of either $u$ or $v$ which is a suffix of $w$.
Let us consider the path from $\varepsilon^{u}$ to $r$. 
%If all of them belong to $Q^{u}$ then the claim is true by the definition of $\mathrsfs{A}_{u}$.
As long as we do not use modified transitions, i.e. the ones that lead from $Q^u$ to $Q^v$ or vice versa, the claim holds true by the definition of $\mathrsfs{A}'_{u}$ and $\mathrsfs{A}'_{v}$. Suppose now that the path contains a transition $u[1 .. n - 1] \xrightarrow{u[n]} s$ and 
$\varepsilon^{u} \dt w' = u[1..n-1]$, where $w'$ is a prefix of $w$. 
Let $s'$ be the maximal prefix of the word $v$ which is a suffix of $w'u[n]$. Note that $w'u[n]$ has $u$ as a suffix. Therefore $|s'| < |u|$,
otherwise $u$ is a factor of $v$. Since $|s'| < |u|$ we have $s' = s$.
Similar reasoning applies in case of transition $v[1 .. m - 1] \xrightarrow{v[m]} p$. Therefore, the claim holds true.
It is not hard to see that also $\varepsilon^{v} \dt w = r'$, where $r'$ is the maximal prefix of either $u$ or $v$ which is a suffix of $w$.

Now we are ready to show that $u$ is a synchronizing word for $\mathrsfs{D}_{u,v}$. By the definition of $\mathrsfs{D}_{u,v}$ we have $\varepsilon^{u} \dt u = s$.
Let us consider an arbitrary state $t \in Q^u$.
Let $\varepsilon^{u} \dt tu = r$. Note that the maximal prefix of the word $u$ which is a suffix 
of $tu$ is equal to $u$. Then by the claim $r$ is a prefix of $v$. Since $u$ is not a factor of $v$ we have $|r| < |u|$. Thus, $r = s$ due
to maximality. Since $\varepsilon^{u} \dt t = t$ we have $t \dt u = s$. Thus, $Q^u \dt u = \{s\}$. Arguing in the same way for the state
$\varepsilon^v$ we get $Q^v \dt u = \{s\}$. So, $u$ is synchronizing. Analogously, one can show that $v$ is synchronizing.

To complete the proof it remains to verify that each word from the set $\Syn(\mathrsfs{D}_{u,v})$ contains $u$ or $v$ as a factor. Take $w\in \Syn(\mathrsfs{D}_{u,v})$ and a state $r\in Q^u$. If $r\dt w\in Q^u$, then $w$ maps all states in the component $Q^v$ into the same state. In particular, $\varepsilon^v\dt w\in Q^u$. Thus $v$ appears in $w$ as a factor. Analogously if $r\dt w\in Q^v$, the word $u$ appears as a factor in $w$. So we proved that $\Syn(\mathrsfs{D}_{u,v})=\Sigma^{*}(u+v)\Sigma^{*}$.
\qed
\end{proof}

\textbf{Acknowledgement}. The authors acknowledge support from the Presidential Programm for young researchers, grant MK-266.2012.1.

\end{document}